\documentclass[12pt,a4paper]{article}
\usepackage[centertags]{amsmath}
\usepackage{amsfonts,amsthm,amssymb}
\usepackage{amssymb}
\usepackage{amsmath}
\usepackage{graphicx}
\usepackage{bbm}
\usepackage{color}
\usepackage{mathtools}

\def\E{\mathbb{E}}
\def\P{\mathbb{P}}
\def\poly{poly}

\theoremstyle{definition}
\newtheorem{theorem}{Theorem}
\newtheorem{corollary}{Corollary}

\newtheorem{lemma}{Lemma}

\newtheorem{example}{Example}
\newtheorem{question}{Question}

\begin{document}

\title{Approximate Nash Equilibria via Sampling}
\author{Yakov Babichenko\footnote{Center for the Mathematics of Information, and Department of Computing and Mathematical Sciences, California Institute of Technology. E-mail:babich@caltech.edu.}, Ron Peretz\footnote{Department of Mathematics, London School of Economics. E-mail: ronprtz@gmail.com}}


\maketitle

\begin{abstract}
We prove that in a normal form $n$-player game with $m$ actions for each player, there exists an approximate Nash equilibrium where each player randomizes uniformly among a set of $O(\log m + \log n)$ pure strategies. This result induces an $N^{\log \log N}$ algorithm for computing an approximate Nash equilibrium in games where the number of actions is polynomial in the number of players ($m=poly(n)$), where $N=nm^n$ is the size of the game (the input size).

In addition, we establish an inverse connection between the entropy of Nash equilibria in the game, and the time it takes to find such an approximate Nash equilibrium using the random sampling algorithm. 
\end{abstract}

\section{Introduction}

Sampling from a Nash equilibrium is a well-know method for proving existence of a simple approximate Nash equilibrium. By the sampling method, the (possibly complicated) mixed strategy $x_i$ of player $i$ is replaced by $k$ i.i.d. samples of pure strategies from the distribution $x_i$. These $k$ samples are each chosen at random with probability $1/k$, and together they form a simple $k$\emph{-uniform strategy} $s_i$. Equivalently, $k$-uniform strategies are mixed strategies that assign to each pure strategy a rational probability with denominator $k$. The main advantage of the $k$-uniform strategy $s_i$ over the original strategy $x_i$ is that there are at most $m^k$ such strategies (actually $\binom{m+k-1}{k}$), where $m$ is the number of actions of player $i$. Therefore, in the case where we do not know the original strategy $x_i$ (and thus we cannot produce the strategy $s_i$ from $x_i$), we can \emph{search} for the strategy $s_i$ over a relatively small set of size $m^k$. 

The sampling method has a very important consequence for the computation of approximate Nash equilibria. If we prove existence of a $k$-uniform approximate Nash equilibrium $(s_i)_{i=1}^n$ for small $k$, then we need only search exhaustively for an approximate Nash equilibrium over all the possible $n$-tuples of $k$-uniform strategies. Although this method seems naive, it provides the best upper bound that is known today for computing an approximate Nash equilibrium.

Althofer \cite{A} was the first to introduce the sampling method, when he studied two-player zero-sum games and showed existence of $k$-uniform approximately optimal strategies with $k=O(\log m)$. Althofer \cite{A} also showed that the order of $\log m$ is optimal (for two-player games). Lipton, Markakis, and Mehta \cite{LMM} generalized this result to all two-player games; i.e., they proved existence of a $k$-uniform approximate Nash equilibrium for $k=O(\log m)$. For $n$-player games, Lipton, Markakis, and Mehta \cite{LMM} proved existence of a $k$-uniform approximate Nash equilibrium for $k=O(n^2 \log m)$. H\'emon, Rougemont, and Santha \cite{HRS} simplified it to $k=O(n \log m)$. 

In the present paper, we prove existence of a $k$-uniform approximate Nash equilibrium for $k=O(\log n + \log m)$ (see Theorem \ref{theo:main}). The results in \cite{LMM} and \cite{HRS} induce a $\poly(N^{\log N})$ algorithm for computing an approximate Nash equilibrium (see \cite{N}), where $N=nm^n$ is the input size. Our result yields a $\poly(N^{\log \log N})$ algorithm for games where the number of actions of each player is polynomial in $n$ (the number of players). To our knowledge, the best previously known upper bound for this class of games is the $\poly(N^{\log N})$ of \cite{LMM}.

Our second result establishes an inverse connection between the entropy of Nash equilibria in the game and the time that it takes the sampling method algorithm to find an approximate Nash equilibrium (see Theorem \ref{theo:ent}). In particular, this result generalizes the result of Daskalakis and Papadimitriou \cite{DP} on existence of a polynomial algorithm for an approximate Nash equilibrium in \emph{small probability games}, which are a sub-class of the games where the entropy of a Nash equilibrium is very high. Daskalakis and Papadimitriou \cite{DP} proved this result for two-player games. A corollary of our result (see Corollary \ref{cor:small}) is that an appropriate generalization of that statement holds for any number of players $n$.

\section{The results}

We consider $n$-player games with $m$-actions for each player.\footnote{All the results in the paper hold also for the case where each player has a different number of actions (i.e., player $i$ has $m_i$ actions). For simplicity, we assume throughout that all players have the same number of actions $m$.} The \emph{size of the game} is denoted by $N:=nm^n$. We use the following standard notation. The set of players is $[n]=\{1,2,...,n\}$. The set of actions of each player is $A_i=[m]=\{1,2,...,m\}$. The set of strategy profiles is $A=[m]^n$. The payoff function of player $i$ is $u_i:A\rightarrow [0,1]$. The payoff function profile is denoted by $u=(u_i)_{i\in [n]}$. The set of probability distributions over a set $B$ is denoted by $\Delta(B)$. The set of mixed actions of player $i$ is $\Delta(A_i)$. The payoff function can be multilinearly extended to $u_i:\Delta(A)\rightarrow [0,1]$. 

A mixed action profile $x=(x_i)_{i\in [n]}$, where $x_i \in \Delta(A_i)$ is an $\varepsilon$-\emph{equilibrium} if no player can gain more than $\varepsilon$ by a unilateral deviation; i.e., $u_i(x)\geq u_i(a_i,x_{-i})-\varepsilon$, for every player $i$ and every action $a_i\in [m]$, where $x_{-i}$ denotes the action profile of all players other than $i$. A $0$-equilibrium is called an \emph{exact} or \emph{Nash} equilibrium.

A mixed strategy $x_i\in A_i$ is called $k$\emph{-uniform} if $x_i(a_i)=c_i/k$, where $c_i\in \mathbb{Z}$, for every action $a_i\in A_i$. Equivalently, a $k$-uniform strategy is a uniform distribution over a multi-set of $k$ pure actions. A strategy profile $x=(x_i)_{i\in [n]}$ will be called $k$\emph{-uniform} if every $x_i$ is $k$-uniform.

We use the notation $f(x)=\poly(g(x))$ if there exists a constant $c$ such that $f(x)\leq g(x)^{c}$ for large enough $x$.

\subsection{General games}
Our Main Theorem states the following:
\begin{theorem}\label{theo:main}
Every $n$-player game with $m$ actions for each player admits a $k$-uniform $\varepsilon$-equilibrium for every
\begin{equation*}
k\geq \frac{8(\ln m + \ln n - \ln \varepsilon + \ln 8)}{\varepsilon^2}.
\end{equation*}. 
\end{theorem}

\begin{corollary}\label{cor:main}
Let $m=\poly(n)$, and let $N=nm^n$ be the input size of an $n$-player $m$-action normal-form game. For every constant $\varepsilon>0$, there exists an algorithm for computing an $\varepsilon$-equilibrium in $\poly(N^{\log \log N})$ steps.
\end{corollary}

\begin{proof}[Proof of Corollary \ref{cor:main}]
The number of all the possible $k$-uniform profiles is at most $m^{nk}$. Note that
\begin{equation*}
m^{nk}=\poly(m^{n \log n})=\poly((m^n)^{\log \log (m^n)})=\poly(N^{\log \log N}). 
\end{equation*}
Therefore the exhaustive search algorithm that searches for an $\varepsilon$-equilibrium over all possible $k$-uniform profiles finds such an $\varepsilon$-equilibrium after at most $\poly(N^{\log \log N})$ iterations.
\end{proof}

\begin{proof}[Proof of Theorem \ref{theo:main}]
The proof uses the sampling method. Let \linebreak
$k \geq \frac{8(\ln m + \ln n - \ln \varepsilon + \ln 8)}{\varepsilon^2}$, and let $x=(x_i)_{i\in [n]}$ be an exact equilibrium of the game $u=(u_i)_{i\in [n]}$. For every player $i$, we sample $k$ i.i.d. pure strategies $(b^i_j)_{j \in k}$ according to the distribution $x_i$ ($b^i_j \in A_i$). Denote by $s_i$ the uniform distribution over the pure actions $(b^i_j)_{j \in k}$. It is enough to show that with positive probability the profile $(s_i)_{i\in [n]}$ forms an $\varepsilon$-equilibrium.

For every player $i$ and strategy $j\in A_i=[m]$, we define a set of forbidden $s$ values:
\begin{equation*}
E_{i,j}=\{\mathbf{s}\in \bigtimes_{l\in[n]}\Delta(A_l):|u_i(j,x_{-i})-u_i(j,\mathbf{s_{-i}})|\geq \frac{\varepsilon}{2}\}.
\end{equation*}

Note that almost every realization of $s$ is absolutely continuous with respect to $x$, written $s\ll x$; i.e., the event $\{\mathrm{support}(s)\subset \mathrm{support}(x)\}$ has probability 1. Therefore, it is sufficient to verify that $\P(s\notin\cup_{i,j} E_{i,j})>0$, since every strategy profile $\mathbf s\ll x$, $\mathbf{s}\notin \cup_{i,j} E_{i,j}$ is an $\varepsilon$-equilibrium, by
\begin{multline*}
u_i(a_i,\mathbf{s_{-i}}) \leq u_i(a_i,x_{-i})+\frac{\varepsilon}{2}\leq \sum_{b\in A_i} \mathbf{s_i}(b) u_i(b,x_{-i}) +\frac{\varepsilon}{2}\\
\leq \sum_{b\in A_i} \mathbf{s_i}(b) u_i(b,\mathbf{s_{-i}}) +\varepsilon = u_i(\mathbf s)+\varepsilon,
\end{multline*}
where the second inequality holds because all the strategies in the support of $\mathbf{s_i}$ are in the support of $x_i$, which contains only best replies to $x_{-i}$.

To show that $\P (s\in\cup_{i,j} E_{i,j})<1$, it is sufficient to show that $\P(s\in E_{i,j})\leq \frac{1}{mn}$ because we have $mn$ such events $\{s\in E_{i,j}\}$.

Up to this point, the arguments of the proof are similar to \cite{LMM} and \cite{HRS}. The estimation of the probability $\P(s\in E_{i,j})$, however, uses more delicate arguments. Let us estimate $\P(s\in E_{1,1})$.

We begin by rewriting the payoff of player 1. For every $l\in[k]$, we can write
\begin{equation*}
u_1(1,s_{-1})=\frac{1}{k^{n-1}}\underset{j_1,j_2,...,j_n \in [k]}{\sum} u_1(1,b^2_{j_2+l},b^3{j_3+l},...,b^n_{j_n+l}) 
\end{equation*}
where the indexes $j_i+l$ are taken modulo $k$. If we take the average over all possible $l$ we have
\begin{equation}\label{eq:pay}
u_1(1,s_{-1})=\frac{1}{k^{n-1}}\underset{j_1,j_2,...,j_n \in [k]}{\sum} \frac{1}{k} \underset{l\in[k]}{\sum} u_1(1,b^2_{j_2+l},b^3{j_3+l},...,b^n_{j_n+l}). 
\end{equation}
For every initial profile of indexes $j_*=(j_2,j_3,...,j_n)\in [k]^{n-1}$ and every $l\in [k]$, we denote $b^{-1}_{j_*+l}:=(b^2_{j_2+l},b^3_{j_3+l},...,b^n_{j_n+l})\in A_{-1}$, and we define the random variable 
\begin{equation}\label{eq:d}
d(j_*):= \begin{cases}
0 & \text{if } \left\lvert \frac{1}{k}\underset{l\in[k]}{\sum}u_1(1,b^{-1}_{j_*+l})-u_1(1,x_{-1}) \right\rvert \leq \dfrac{\varepsilon}{4} \\
1 & \text{otherwise.}
\end{cases}
\end{equation}
By the definition of $d(j_*)$, we have
\begin{equation}\label{eq:din}
d(j_*)+\frac{\varepsilon}{4} \geq \left\lvert \frac{1}{k}\underset{l\in[k]}{\sum}u_1(1,b^{-1}_{j_*+l})-u_1(1,x_{-1}) \right\rvert.
\end{equation}
Note also that for any fixed $j_*$ the random action profiles $b^{-1}_{j_*+1},\ldots,b^{-1}_{j_*+k}$ are independent. Therefore by Hoeffding's inequality (see \cite{H}) we have
\begin{equation}\label{eq:hof}
\E(d(j_*))\leq 2 e^{-\frac{k\varepsilon^2}{8}}.
\end{equation}
Using representation (\ref{eq:pay}) of the payoffs and inequalities (\ref{eq:din}) and (\ref{eq:hof}), we get
\begin{equation}\label{eq:fin}
\begin{aligned}
\P (s\in E_{1,1}) &= \P \left( \left\lvert \frac{1}{k^{n-1}} \underset{j_* \in [k]^{n-1}}{\sum} \frac{1}{k} \underset{l\in[k]}{\sum} u_1(1,b^{-1}_{j_*+l})-u_1(1,x_{-1}) \right\rvert \geq \frac{\varepsilon}{2} \right) \\
&\leq \P \left( \frac{1}{k^{n-1}} \underset{j_* \in [k]^{n-1}}{\sum} \left\lvert \frac{1}{k} \underset{l\in[k]}{\sum} u_1(1,b^{-1}_{j_*+l})-u_1(1,x_{-1}) \right\rvert \geq \frac{\varepsilon}{2} \right) \\
&\leq \P \left( \frac{1}{k^{n-1}} \underset{j_* \in [k]^{n-1}}{\sum} d(j_*) \geq \frac{\varepsilon}{4} \right) \leq \frac{8 e^{-\frac{k\varepsilon^2}{8}}}{\varepsilon}
\end{aligned}
\end{equation}
where the last inequality follows from Markov's inequality.
Putting $k\geq \frac{8(\ln m + \ln n - \ln \varepsilon + \ln 8)}{\varepsilon^2}$ in inequality (\ref{eq:fin}), we get $\P (E_{1,1}) \leq \frac{1}{mn}$.
\end{proof}

\subsection{Games with a high-entropy equilibrium}

In the sequel it will be convenient to consider the set of $k$-uniform strategies as the set of \emph{ordered} $k$-tuples of pure actions. To avoid ambiguity we will call those strategies $k$\emph{-uniform ordered strategies}.\footnote{Many $k$-uniform ordered strategies correspond to the same mixed strategy of the player in the game.} Now the number of $k$-uniform ordered profiles is exactly $m^{nk}$.

The algorithm of Corollary \ref{cor:main} suggests that we should search over all the possible $k$-uniform profiles (or $k$-uniform ordered profiles), one by one, until we find an approximate equilibrium. Consider now the case where a large fraction of the $k$-uniform ordered strategies form an approximate equilibrium, say a fraction of $1/r$. In such a case we can pick $k$-uniform ordered profiles \emph{at random}, and then we will find the approximate equilibrium in expected time $r$.

Define the \emph{$k$-uniform random sampling algorithm} ($k$-URS) to be the algorithm described above; i.e., it samples uniformly at random $n$-tuples of $k$-uniform ordered strategies and checks whether this profile forms an $\varepsilon$-equilibrium.\footnote{Checking whether a strategy profile forms an approximate equilibrium can always be done in $\poly(N)$ time. Actually, it can even be done by using only $\poly(n,m)$ samples from the mixed profile. Using the samples, the answer will be correct with a probability that is exponential (in $n$ and $m$) close to 1 (see, e.g., \cite{B}, proof of Theorem 2).}

An interesting question arises: For which games does the $k$-URS algorithm find an approximate equilibrium fast? Daskalakis and Papadimitriou \cite{DP} focused on two-player games with $m$ actions, and they showed that the $k$-URS algorithm finds an approximate equilibrium after $\poly(m)$ samples for \emph{small-probability games}. A \emph{small-probability game} is a game that admits a Nash equilibrium where each pure action is played with probability at most $c/m$ for some constant $c$.

Here we generalize the result of Daskalakis and Papadimitriou to $n$-player games. Instead of focusing on the specific class of small-probability games we establish a general connection between the entropy of equilibria in the game and the expected number of samples of the $k$-URS algorithm until an approximate Nash equilibrium is found.

\begin{theorem}\label{theo:ent}
Let $u$ be an $n$-player game with $m$ actions for each player, with a Nash equilibrium $x=(x_i)$. Let $k\geq \max\{\frac{16}{\varepsilon^2}(\ln n + \ln m -\ln \varepsilon +2), e^{16/\varepsilon^2}\}=O(\log m + \log n)$; then the $k$-uniform random sampling algorithm finds an $\varepsilon$-equilibrium after at most $4\cdot 2^{k(n\log_2 m -H(x))}$ samples in expectation, where $H(x)$ is Shannon's entropy of the Nash equilibrium $x$.
\end{theorem}

The following corollary of this theorem is straightforward.

\begin{corollary}\label{cor:ent}
Families of games where $n\log_2 m - \max\limits_{x\in \mathrm{NE}}H(x)$ is bounded admit a $poly(m,n)$ probabilistic algorithm for computing an approximate Nash equilibrium.
\end{corollary}

The corollary follows from the fact that $k=O(\log m + \log n)$, and therefore $4\cdot 2^{kO(1)} = \poly(n,m)$. 

A special case where $n\log_2 m-H(x)$ is constant is that of small-probability games with a \emph{constant} number of players $n$.

\begin{corollary}\label{cor:small}
Let $c\geq 1$, and let $u$ be an $n$-player $m$-action game with a Nash equilibrium $x=(x_i)_{i\in [n]}$, where $x_i(a_i) \leq \frac{c}{m}$ for players $i$ and all actions $a_i\in A_i$. Let $k=O(\log m)$, as defined in Theorem \ref{theo:ent}. Then the expected number of samples of the $k$-URS algorithm is at most $4\cdot 2^{k n \log c}=\poly(m)$. 
\end{corollary}

The corollary follows from the fact that the entropy of the Nash equilibrium $x$ is $H(x)=\sum_{i\in [n]} H(x_i) \geq n(\log_2 m-\log_2 c)$.

The following example demonstrates that even in the case of two-player games, the class of games that have PTAS according to Corollary \ref{cor:ent} is slightly wider than the class of small-probability games.

\begin{example}
Consider a two-player $m$-action game where the equilibrium is $x=(x_1,x_2)$, where $x_1$ is the uniform distribution over all actions $x_1=(\frac{1}{m},\frac{1}{m},...,\frac{1}{m})$, and $x_2=(\frac{1}{\sqrt{m}},\frac{1}{m+\sqrt{m}},\frac{1}{m+\sqrt{m}},...,\frac{1}{m+\sqrt{m}})$. This game is not a small-probability game, but it does satisfy $n\log_2 m-H(x)=o(1)$: 
\begin{eqnarray*}
2\log_2 m - H(x) &\leq & \log_2 m - \frac{m-1}{m+\sqrt{m}}\log_2 (m+\sqrt{m}) \\
 &\leq & \frac{1}{\sqrt{m}+1} \log m = o(1).
\end{eqnarray*}
\end{example}

In the proof of Theorem \ref{theo:ent} we use the following lemma from information theory.

\begin{lemma}\label{lem:inf}
Let $y$ be a random variable that assumes values in a finite set $M$. Let $S\subset M$ such that $\P (y\in S)\geq 1-\frac{1}{\log_2 |M|}$; then $|S|\geq \frac{1}{4} 2^{H(y)}$.
\end{lemma}

\begin{proof}
\begin{equation*}
\begin{aligned}
H(y) &=\P (y\in S) H(y|y\in S) + \P (y\notin S) H(y|y\notin S) + H(\mathbbm{1}_{\{y\in S\}}) \\
&\leq \log_2 |S| + \P (y\notin S) \log_2 |M| +1 \leq \log_2 |S| + 2.
\end{aligned}
\end{equation*}
\end{proof}

\begin{proof}[Proof of Theorem \ref{theo:ent}]
Note that $k\geq\max\{\frac{16}{\varepsilon^2}(\ln n + \ln m -\ln \varepsilon +2), e^{16/\varepsilon^2}\}$ guarantees that
\begin{equation*}
\frac{8 e^{-\frac{k\varepsilon^2}{8}}}{\varepsilon} \leq \frac{1}{mn} \frac{1}{nk log_2 m}.
\end{equation*}
By considering inequality (\ref{eq:fin}) in the proof of Theorem \ref{theo:main}, we can see that the above choice of $k$ implies that $\P (E_{1,1}) \leq \frac{1}{mn} \frac{1}{nk \log_2 m}$, which implies that $\P(s\in\cup_{i,j} E_{i,j}) \leq \frac{1}{nk \log_2 m}$. This means that if we sample $k$-uniform ordered strategy profiles according to the Nash equilibrium $x$, then the resulting $k$-uniform ordered strategies form an $\varepsilon$-equilibrium with a probability of at least $1-\frac{1}{nk \log_2 m}=1-\frac{1}{\log_2(m^{nk})}$.

Next, using Lemma \ref{lem:inf}, we provide a lower bound on the number of $k$-uniform profiles that form an $\varepsilon$-equilibrium. The random $k$-uniform profiles are elements of a set of size $m^{nk}$. The entropy of the random $k$-uniform profile is $kH(x)$. The probability that the random profile will form an $\varepsilon$-equilibrium is at least $1-\frac{1}{\log_2(m^{nk})}$. Therefore, by Lemma \ref{lem:inf}, we get that there are at least $\frac{1}{4}2^{kH(x)}$ different $k$-uniform profiles that are $\varepsilon$-equilibria.

To conclude, the fraction of the $k$-uniform profiles that form an $\varepsilon$-equilibrium (among all the $k$-uniform profiles) is at least:
\begin{equation*}
\frac{\frac{1}{4}2^{kH(x)}}{m^{nk}}=\frac{1}{4}2^{k(H(x)-n\log_2 m)}.
\end{equation*}
Therefore, the expected time for finding an $\varepsilon$-equilibrium is at most $4\cdot 2^{k(n\log_2 m -H(x))}$.
\end{proof}

\section{Discussion}
Having established an upper bound of $O(\log m + \log n)$, it is natural to ask whether it is tight. Althofer \cite{A} provided a lower bound of the order $\log m$ that matches our upper bound in the case where the number of players is not much larger than the number of pure strategies; i.e., $n=\poly(m)$. In general, the tightness of our upper bound remains an open question. A similar question regarding the existence of \emph{pure} approximate equilibria in \emph{Lipschitz} games with many players arose in a related work by Azrieli and Shmaya \cite{AS}.

Let us call games with $n$ players, $m$ actions for each player, and payoffs in $[0,1]$, \emph{normalized $n$-player $m$-action games}. To pinpoint the limits of our understanding of the problem, consider the following questions.

\begin{question}\label{q untight}
Is there a function $k\colon\mathbb (0,1)\to\mathbb N$ ($k$ dependents on $\varepsilon$ only, and not on the number of players $n$), such that every normalized $n$-player two-action game admits an $\varepsilon$-equilibrium in which every player employs a mixed strategy whose coefficients are rational numbers with a denominator at most $k(\varepsilon)$?  
\end{question}

\begin{question}\label{q tight}
Is there an $\varepsilon>0$ and a constant $C>0$, such that for every $n,m\in \mathbb N$ there exists a normalized $n$-player $m$-action game that does not admit any $\varepsilon$-equilibrium in which every player employs a mixed strategy whose coefficients are rational numbers with a denominator at most $C(\log n+\log m)$?
\end{question}

Note that a positive answer to Question~\ref{q tight} means that our upper bound \emph{is} tight, whereas a positive answer to Question~\ref{q untight} implies that our upper bound is \emph{not} tight. A positive answer to Question~\ref{q untight} means that one can find a $k$-uniform approximate equilibrium of the game for a \emph{constant} $k$ (depending only on $\varepsilon$), which in particular implies that there exists a $poly(N)$ algorithm for computing an approximate Nash equilibrium in two-action games.

\end{document}